\documentclass[journal,onecolumn]{IEEEtran}

\usepackage{amsmath,amssymb,amsthm,mathtools}
\usepackage{bm}
\usepackage{graphicx}
\usepackage{booktabs}
\usepackage{enumitem}
\usepackage{hyperref}
\usepackage{xcolor}
\usepackage{listings}
\usepackage{cite}

\theoremstyle{plain}
\newtheorem{theorem}{Theorem}
\newtheorem{lemma}{Lemma}

\newtheorem{proposition}{Proposition}
\theoremstyle{remark}
\newtheorem{remark}{Remark}

\newcommand{\R}{\mathbb{R}}
\newcommand{\E}{\mathbb{E}}

\newcommand{\diag}{\operatorname{diag}}

\newcommand{\logtwo}{\log_2}

\title{On the Entropy of General Mixture Distributions}

\author{Namyoon Lee,~\IEEEmembership{Senior Member,~IEEE}%
\thanks{N. Lee is with the Department of Electrical Engineering, POSTECH, Pohang, South Korea
(e-mail: \texttt{nylee@postech.ac.kr}).}%
}

\begin{document}
\maketitle

\begin{abstract}
Mixture distributions are a workhorse model for multimodal data in information theory, signal processing, and machine learning. Yet even when each component density is simple, the differential entropy of the mixture is notoriously hard to compute because the mixture couples a logarithm with a sum. This paper develops a deterministic, closed-form toolkit for bounding and accurately approximating mixture entropy directly from component parameters. Our starting point is an information-theoretic channel viewpoint: the latent mixture label plays the role of an input and the observation is the output. This viewpoint separates mixture entropy into an average within-component uncertainty plus an overlap term that quantifies how much the observation reveals about the hidden label. We then bound and approximate this overlap term using pairwise overlap integrals between component densities, yielding explicit expressions whenever these overlaps admit closed form. A simple, family-dependent offset corrects the systematic bias of the Jensen overlap bound and is calibrated to be exact in the two limiting regimes of complete overlap and near-perfect separation. A final clipping step guarantees that the estimate always respects universal information-theoretic bounds. Closed-form specializations are provided for Gaussian, factorized Laplacian, uniform, and hybrid mixtures, and numerical experiments validate the resulting bounds and approximations across separation, dimension, number of components, and correlated covariances.
\end{abstract}


\begin{IEEEkeywords}
Differential entropy, mixture distributions, Gaussian mixture model, Laplacian mixture, uniform mixture,
entropy bounds, Jensen's inequality, mutual information, collision entropy, Monte Carlo estimation.
\end{IEEEkeywords}

\section{Introduction}
\label{sec:intro}

Mixture distributions provide a parsimonious model for the ubiquitous situation in which the observable random vector
is continuous, while its generative mechanism includes a discrete \emph{mode}, \emph{cluster}, or \emph{state}.
A $K$-component mixture has density
\begin{equation}
\label{eq:mix_pdf}
p_X(x)=\sum_{c=1}^K \pi_c f_c(x),\qquad
\pi_c>0,\ \sum_{c=1}^K \pi_c=1,
\end{equation}
where $C\in\{1,\dots,K\}$ is a latent label with $\mathbb{P}(C=c)=\pi_c$ and $X\mid(C=c)\sim f_c$.
Such models are foundational in statistics and machine learning as universal density approximators and
as the basis of model-based clustering and discriminant analysis
\cite{BanfieldRaftery1993ModelBased,McLachlanPeel2000Mixtures,FraleyRaftery2002ModelBasedClustering},
and are typically fit using maximum-likelihood methods such as the expectation-maximization (EM) algorithm \cite{DempsterLairdRubin1977EM}.
Beyond Gaussian mixtures, heavy-tailed components are used to capture robustness and sparsity:
mixtures of shifted asymmetric Laplace distributions provide flexible tail and skewness modeling in high-dimensional
data analysis \cite{FranczakBrowneMcNicholas2014SALMixtures}, and early empirical studies already used Laplace--normal
(Gaussian--Laplacian) mixtures for wind-shear data \cite{JonesMcLachlan1990LaplaceNormal}.
On the other hand, bounded-support components arise whenever uncertainty is constrained by quantization,
saturation, or feasibility constraints; mixtures of uniform distributions have been studied both for parameter
estimation in finite uniform mixtures \cite{CraigmileTitterington1997UniformMixtures} and for special cases such as
two-component uniform mixtures \cite{HusseinLiu2009TwoUniforms}.
Hybrid mixtures combining nominal and outlier models are also common:
Gaussian mixtures augmented with a uniform  noise component are a classical robust clustering device
\cite{CorettoHennig2011GaussianUniform} and have been further analyzed under separation constraints
\cite{Coretto2022GaussianUniformNoiseSeparation}, while Gaussian--Laplacian hybrids provide a simple dense-plus-impulsive
modeling primitive that also appears in applied settings \cite{ShenoyGorinevsky2014GaussianLaplacian}.

A central information-theoretic quantity associated with \eqref{eq:mix_pdf} is the differential entropy
\begin{equation}
\label{eq:intro_h_def}
h(X)=-\int_{\R^n} p_X(x)\logtwo p_X(x)\,dx,
\end{equation}
which governs lossless compression limits and appears throughout information theory \cite{CoverThomas2006}.
Despite the tractability of many component families, mixture entropy rarely admits a closed form: the mixture introduces
a log-of-sum term $\log(\sum_c \pi_c f_c(x))$ inside the expectation, which is the analytic bottleneck.


The key viewpoint of this paper is to treat \eqref{eq:mix_pdf} as a memoryless channel with input $C$ and output $X$.
This yields the exact decomposition
\begin{equation}
\label{eq:intro_decomp}
h(X)=h(X\mid C)+I(X;C),
\end{equation}
where $h(X\mid C)=\sum_{c=1}^K \pi_c h(f_c)$ is the mixture-weighted component entropy and $I(X;C)$ is the mutual
information between continuous $X$ and discrete $C$.
Thus mixture entropy evaluation reduces to estimating the information term $I(X;C)$, which captures the overlap geometry
of the component densities.
Since $0\le I(X;C)\le H(C)$ for discrete $C$, \eqref{eq:intro_decomp} implies the universal label-sandwich bounds
\begin{equation}
\label{eq:intro_sandwich}
h(X\mid C)\ \le\ h(X)\ \le\ h(X\mid C)+H(C).
\end{equation}
These bounds are sharp at two extremes: complete overlap ($I(X;C)=0$) and vanishing overlap ($I(X;C)\approx H(C)$).
Most practical regimes lie in between, motivating explicit bounds and accurate approximations that remain computable
from component parameters.

\subsection{Related work}
\label{subsec:related_work}

Finite mixtures are a standard tool for density estimation, clustering, and latent-variable modeling;
see, e.g., the monograph \cite{McLachlanPeel2000Mixtures} and the foundational EM framework
\cite{DempsterLairdRubin1977EM}.
Model-based clustering with Gaussian and non-Gaussian components is studied classically in
\cite{BanfieldRaftery1993ModelBased,FraleyRaftery2002ModelBasedClustering}.
Beyond purely Gaussian components, Laplace/Laplace-normal mixtures arise in applied modeling and robust
statistics \cite{JonesMcLachlan1990LaplaceNormal,FranczakBrowneMcNicholas2014SALMixtures}, while uniform
mixtures appear in bounded-support uncertainty modeling and parameter estimation
\cite{CraigmileTitterington1997UniformMixtures,HusseinLiu2009TwoUniforms}.
Hybrid mixtures that explicitly combine Gaussian structure with bounded noise or outlier components
are developed in \cite{CorettoHennig2011GaussianUniform,Coretto2022GaussianUniformNoiseSeparation}, and
Gaussian--Laplacian hybrids appear in application-driven robust modeling
\cite{ShenoyGorinevsky2014GaussianLaplacian}.

Deterministic approximations and bounds for Gaussian mixture entropy have a long history.
Huber \emph{et al.}~\cite{HuberBaileyDurrantWhyteHanebeck2008} develop practical approximations and
computable bounds (including Jensen-type bounds that move the logarithm outside the intractable integral).
Moshksar and Khandani~\cite{MoshksarKhandani2016EntropyGM} derive arbitrarily tight upper and lower bounds on
the differential entropy of Gaussian mixtures.
More recently, Dahlke and Pacheco~\cite{DahlkePacheco2023} analyze polynomial (Taylor/Legendre) entropy
approximations, identify regimes where widely used expansions can diverge, and propose convergent series
constructions for Gaussian mixture model (GMM) entropy.

Closely related to mixture entropy evaluation are tractable surrogates for divergence and cross-entropy
quantities between mixtures, which arise in model comparison and variational objectives.
Hershey and Olsen~\cite{HersheyOlsen2007} propose efficient approximations for the Kullback-Leibler (KL) divergence between
Gaussian mixture models; related approximation ideas also appear in application contexts such as image
similarity for Gaussian mixtures \cite{GoldbergerGordonGreenspan2003KL2GMM}.

A complementary viewpoint is to control the fundamental analytic obstacle directly: the log-sum-exp (LSE)
structure of mixture densities.
Nielsen and Sun~\cite{NielsenSun2016} develop guaranteed (certifiable) bounds for information-theoretic
functionals of univariate mixtures via piecewise LSE inequalities, which can be tightened systematically
by refining the partition.

A common heuristic is the separated-components rule $h(X)\approx h(X\mid C)+H(C)$, which becomes accurate when overlaps are negligible. Furuya \emph{et al.}~\cite{FuruyaKusumotoTaniguchiKannoSuetake2024} provide theoretical error analysis showing how approximation quality improves with component separation.

From an information-theoretic viewpoint, the obstruction to evaluating the mixture entropy
$h\!\big(\sum_c \pi_c f_c\big)$ in closed form is exactly the \emph{concavity deficit} of differential entropy,
which can be interpreted as a Jensen--Shannon-type divergence among the mixture components
\cite{Lin1991JSD}. This perspective places mixture-entropy bounds within the broader landscape of
information measures and $f$-divergences \cite{Csiszar1967InformationMeasures,SasonVerdu2016FDivergence}.
In a comprehensive and general treatment, Melbourne \emph{et al.}~\cite{MelbourneTalukdarBhabanMadimanSalapaka2022}
derive sharp bounds on this concavity deficit, including mixture-dependent refinements of the classical
upper bound that leverage total-variation proximity between each component and its \emph{mixture complement};
their development connects to the role of total variation in controlling the distribution of relative
information \cite{Verdu2014TV} and draws inspiration from skew-divergence constructions
\cite{Audenaert2014Skew}. Complementary techniques based on symmetric decreasing rearrangements, which yield
entropy inequalities beyond the classical entropy power inequality and illuminate mixture-like structures,
are developed in \cite{WangMadiman2014EPIBeyond}.

Entropy and mutual information objectives are pervasive in experimental design, active learning, and Bayesian optimization:
classical foundations include Lindley~\cite{Lindley1956InformationExperiment} and Bernardo~\cite{Bernardo1979ExpectedInformation},
and maximum-entropy sampling is developed in \cite{ShewryWynn1987MaxEntropySampling}.
Modern Bayesian optimization algorithms explicitly use predictive entropy criteria
\cite{HernandezLobatoHoffmanGhahramani2014PES,RuMcLeodGranziolOsborne2018FIBO}, and variational Bayesian optimal
design is studied in \cite{FosterEtAl2019VBOED}; a broad survey of active learning objectives is given in
\cite{Settles2009ActiveLearningSurvey}.
Information-theoretic feature selection frameworks \cite{BrownPocockZhaoLujan2012CLMFeatureSelection} and
the information bottleneck principle \cite{TishbyPereiraBialek2000InformationBottleneck} further highlight
the demand for tractable entropy and mutual information proxies.

The above literature suggests a recurring trade-off between (i) certified tightness (often achieved via
general divergence metrics) and (ii) closed-form computability for concrete multivariate parametric families.
Our approach emphasizes \emph{pairwise overlap integrals} that admit closed-form expressions for multiple
time-honored mixture families (Gaussian, factorized Laplacian, uniform-on-sets, and hybrid mixtures), and
we propose a \emph{family-calibrated} offset approximation that is exact in the complete-overlap limit
and consistent with universal label-sandwich bounds for \emph{all} mixtures.

\subsection{Contributions}
\label{subsec:contributions}

This paper develops a simple and practical way to reason about—and compute—the differential entropy of mixture
distributions.

\begin{itemize}[leftmargin=*, itemsep=3pt]

\item \textbf{Universal bounds from a channel viewpoint:}
We treat a mixture model as a channel: the discrete component label is the input and the continuous observation is
the output. This immediately separates the mixture entropy into an easy part (the average within-component
uncertainty) and a hard part (how much the observation reveals the label). This viewpoint yields a universal
two-sided bracket on the mixture entropy that holds for any mixture family, in any dimension, without parametric
assumptions.

\item \textbf{A tight and deterministic approximation:}
We replace the hard \textit{label information} term by a closed-form proxy built from pairwise component overlaps. A Jensen step makes this proxy computable, an analytically chosen family-dependent offset fixes the normalization in the complete-overlap regime, and a final clipping step guarantees the estimate never violates the universal bracket. The result is a single approximation that tracks the mixture entropy across overlap regimes while remaining information-theoretically admissible.

\item \textbf{Closed-form formulas for time-honored mixture families:}
We provide explicit plug-in expressions—component entropies, overlap integrals, and offsets—for several
time-honored and widely used mixture classes: Gaussian mixtures, factorized Laplacian mixtures, uniform mixtures on
sets, and representative hybrid mixtures (Gaussian--Laplacian, Gaussian--uniform, Laplacian--uniform) under
standard structural assumptions that preserve closed-form computability.

\item \textbf{Numerical validation across regimes:}
Monte Carlo experiments validate two points: the universal bounds correctly bracket the true mixture entropy, and
the proposed approximation is typically tight across separations, dimensions, and numbers of components. The
experiments also highlight how geometry (compact support versus heavy tails) and dimension control the sharpness of
the transition from \textit{components indistinguishable} to \textit{labels essentially decodable}.

\end{itemize}

\subsection{Organization}
\label{subsec:intro_org}
Section~\ref{sec:bounds} develops universal bounds and tightness conditions.
Section~\ref{sec:approx} introduces the Jensen/overlap bound, offset calibration, and clipped approximation.
Section~\ref{sec:examples} provides closed-form specializations.
Section~\ref{sec:numerics} presents numerical experiments, and Section~\ref{sec:conclusion} concludes.

\section{Universal Bounds via Decomposition}
\label{sec:bounds}
In this section, we introduce a universial bounds of $h(X)$ for general mixture distributions $p_X(x)$. The key idea is to use decomposition of $h(X)$ mixture entropy can be decomposed into a closed-form part plus an  mutual information information between the latent label $C$ and the mixture variable $X$.

\subsection{Entropy Decomposition}
The following lemma formally states the entropy decomposition. 

\begin{lemma}[Entropy decomposition]
\label{lem:decomp}
Let $X\in\R^n$ have mixture density \eqref{eq:mix_pdf} with latent label $C$.
Then the differential entropy decomposes as
\begin{equation}
\label{eq:decomp}
h(X)=h(X\mid C)+I(X;C),
\end{equation}
where $I(X;C)$ is the mutual information (in bits) between continuous $X$ and discrete $C$.
\end{lemma}

\begin{proof}
We prove \eqref{eq:decomp} by writing the mixed joint entropy $h(X,C)$ in two different ways and equating them. The joint law of $(X,C)$ has mixed (continuous--discrete) density
\begin{equation}
\label{eq:joint_density}
p_{X,C}(x,c)=\mathbb{P}(C=c)\,p_{X\mid C}(x\mid c)=\pi_c f_c(x).
\end{equation}
Define the mixed joint entropy
\begin{equation}
\label{eq:joint_entropy_def}
h(X,C)\triangleq -\sum_{c=1}^K \int_{\R^n} p_{X,C}(x,c)\,\logtwo p_{X,C}(x,c)\,dx.
\end{equation}
Substitute $p_{X,C}(x,c)=\pi_c f_c(x)$ into \eqref{eq:joint_entropy_def}:
\begin{align}
h(X,C)
&=-\sum_{c=1}^K \int_{\R^n} \pi_c f_c(x)\,\logtwo\big(\pi_c f_c(x)\big)\,dx \nonumber\\
&=-\sum_{c=1}^K \int_{\R^n} \pi_c f_c(x)\,\big(\logtwo \pi_c+\logtwo f_c(x)\big)\,dx.
\label{eq:joint_entropy_expand}
\end{align}
Split the integral into two parts. For the first term,
\begin{align}
-\sum_{c=1}^K \int_{\R^n} \pi_c f_c(x)\,\logtwo\pi_c\,dx
&=-\sum_{c=1}^K \pi_c\logtwo\pi_c \int_{\R^n} f_c(x)\,dx \nonumber\\
&=-\sum_{c=1}^K \pi_c\logtwo\pi_c = H(C),
\end{align}
since each $f_c$ integrates to $1$. For the second term,
\begin{align}
-\sum_{c=1}^K \int_{\R^n} \pi_c f_c(x)\,\logtwo f_c(x)\,dx 
&=\sum_{c=1}^K \pi_c \Big(-\int_{\R^n} f_c(x)\logtwo f_c(x)\,dx\Big)\nonumber\\
&=\sum_{c=1}^K \pi_c\,h(f_c)=h(X\mid C).
\end{align}
Therefore
\begin{equation}
\label{eq:joint_entropy_H_plus_cond}
h(X,C)=H(C)+h(X\mid C).
\end{equation}
Now we factorize the joint density the other way. We can also write
\begin{equation}
\label{eq:factorize_other}
p_{X,C}(x,c)=p_X(x)\,\mathbb{P}(C=c\mid X=x).
\end{equation}
Substitute \eqref{eq:factorize_other} into \eqref{eq:joint_entropy_def}:
\begin{align}
 h(X,C) 
&=-\sum_{c=1}^K\!\! \int_{\R^n}\! p_X(x)\!\,\mathbb{P}(C=c\mid x)\,
\logtwo\!\!\big(p_X(x)\!\!\,\mathbb{P}(C=c\mid x)\big)\,dx \nonumber\\
&=-\sum_{c=1}^K \int_{\R^n} p_X(x)\,\mathbb{P}(C=c\mid x)\,\logtwo p_X(x)\,dx \nonumber\\
&\quad\ -\sum_{c=1}^K \int_{\R^n} p_X(x)\,\mathbb{P}(C=c\mid x)\,\logtwo \mathbb{P}(C=c\mid x)\,dx.
\label{eq:joint_entropy_expand2}
\end{align}
In the first term of \eqref{eq:joint_entropy_expand2}, $\logtwo p_X(x)$ does not depend on $c$, so
\begin{align}
 -\sum_{c=1}^K \int_{\R^n} p_X(x)\,\mathbb{P}(C=c\mid x)\,\logtwo p_X(x)\,dx 
&=-\int_{\R^n} p_X(x)\Big(\sum_{c=1}^K \mathbb{P}(C=c\mid x)\Big)\logtwo p_X(x)\,dx \nonumber\\
&=-\int_{\R^n} p_X(x)\logtwo p_X(x)\,dx \nonumber\\
&= h(X),
\end{align}
since $\sum_{c}\mathbb{P}(C=c\mid x)=1$. The second term in \eqref{eq:joint_entropy_expand2} equals
\begin{align}
 -\int_{\R^n} p_X(x)\sum_{c=1}^K \mathbb{P}(C=c\mid x)\logtwo \mathbb{P}(C=c\mid x)\,dx =H(C\mid X),
\end{align}
by definition of conditional entropy for discrete $C$ conditioned on continuous $X$. Hence
\begin{equation}
\label{eq:joint_entropy_h_plus_condC}
h(X,C)=h(X)+H(C\mid X).
\end{equation}
By equating \eqref{eq:joint_entropy_H_plus_cond} and \eqref{eq:joint_entropy_h_plus_condC}, we obtain 
\begin{align}
h(X)+H(C\mid X)=H(C)+h(X\mid C).
\end{align}
Finally, rearranging gives
\begin{align}
h(X)&=h(X\mid C)+H(C)-H(C\mid X)\nonumber\\
&=h(X\mid C)+I(X;C),
\end{align}
which completes the proof.
\end{proof}

\vspace{0.1cm}
Leveraging the entropy decomposition lemma, we provide universal bounds for the mixture distribution, stated in the following theorem.

\begin{theorem}[Universal label-sandwich bounds]
\label{thm:label_sandwich}
For any mixture distribution $p_X(x)=\sum_{c=1}^K \pi_c f_c(x)$, the differential entropy is bounded by
\begin{equation}
\label{eq:label_bounds}
h(X\mid C)\ \le\ h(X)\ \le\ h(X\mid C)+H(C).
\end{equation}
\end{theorem}

\begin{proof}
By Lemma~\ref{lem:decomp}, $h(X)=h(X\mid C)+I(X;C)$.
Since mutual information is nonnegative, $I(X;C)\ge 0$, we obtain $h(X)\ge h(X\mid C)$.

Moreover, because $C$ is discrete,
\begin{equation}
I(X;C)=H(C)-H(C\mid X)\le H(C)
\end{equation}
since conditional entropy $H(C\mid X)\ge 0$. Therefore
\begin{equation}
h(X)=h(X\mid C)+I(X;C)\le h(X\mid C)+H(C),
\end{equation}
establishing \eqref{eq:label_bounds}.
\end{proof}


A key implication is that \eqref{eq:label_bounds} is \emph{universal}: it does not rely on any
parametric assumption on the component family $\{f_c\}$ beyond existence/fin\-iteness of the
relevant entropies. In particular, the bound holds for
Gaussian mixtures, Laplacian mixtures, uniform mixtures (bounded support), and hybrid mixtures
(e.g., Gaussian--Laplacian, uniform--Gaussian), as well as mixtures of arbitrary continuous densities.

The reason is structural: the proof uses only the decomposition identity
\[
h(X)=h(X\mid C)+I(X;C),
\]
together with the purely information-theoretic inequality $0\le I(X;C)\le H(C)$ for discrete $C$.
No property of $f_c$ (Gaussianity, log-concavity, bounded support, etc.) is required.

 The interval
\[
\big[h(X\mid C),\; h(X\mid C)+H(C)\big]
\]
can be read as an  uncertainty budget  :
\begin{itemize}[leftmargin=*]
\item $h(X\mid C)$ is the average uncertainty \emph{within} components (intra-component spread).
\item $H(C)$ is the uncertainty of the label itself (inter-component mixing uncertainty).
\item The mixture entropy $h(X)$ differs from $h(X\mid C)$ by exactly $I(X;C)$, i.e.,
      the amount of label information revealed by observing $X$.
\end{itemize}
Thus the mixture entropy is always  conditional entropy plus a penalty   that cannot exceed $H(C)$. Even when $h(X)$ has no closed form, Theorem~\ref{thm:label_sandwich} yields a deterministic bracket as long as $h(X\mid C)$ is computable. For many families, $h(X\mid C)=\sum_c \pi_c h(f_c)$ is closed form, so \eqref{eq:label_bounds} provides immediately computable bounds on $h(X)$.

\subsection{Analysis of the Gap}
\label{subsubsec:gap_identities}

Theorem~\ref{thm:label_sandwich} yields \emph{exact} expressions for the two gaps:
\begin{align}
\label{eq:gap_lower_exact}
h(X)-h(X\mid C) &= I(X;C),\\
\label{eq:gap_upper_exact}
\bigl(h(X\mid C)+H(C)\bigr)-h(X) &= H(C\mid X).
\end{align}
Hence:
\begin{itemize}[leftmargin=*]
\item the \emph{lower-bound gap} equals how informative $X$ is about the label $C$;
\item the \emph{upper-bound gap} equals the residual label ambiguity after observing $X$.
\end{itemize}
These identities are not approximations; they hold exactly for any mixture. Then, an natural question is when these bounds can meet? Because the two gaps \eqref{eq:gap_lower_exact}--\eqref{eq:gap_upper_exact} are information measures, tightness conditions can be stated precisely. The following proposition states these conditions. 

\begin{proposition}[Tightness conditions for the label-sandwich bounds]
\label{prop:tightness_label_sandwich}
Under the mixture model $p_X(x)=\sum_{c=1}^K\pi_c f_c(x)$:
\begin{enumerate}[leftmargin=*, itemsep=2pt]
\item \textbf{Lower bound tightness.}
      The equality $h(X)=h(X\mid C)$ holds if and only if
      \begin{equation}
      \label{eq:lower_tight_iff}
      I(X;C)=0 \Longleftrightarrow X \perp C \Longleftrightarrow f_c(x)=p_X(x)\ \text{a.e. for all }c.
      \end{equation}
\item \textbf{Upper bound tightness.}
      The equality $h(X)=h(X\mid C)+H(C)$ holds if and only if
      \begin{equation}
      \label{eq:upper_tight_iff}
      H(C\mid X)=0 \Longleftrightarrow \exists\ \text{measurable }g:\R^n\to\{1,\dots,K\}
      \end{equation}
      such that $C=g(X)$ almost surely. 
\item \textbf{Bounds meet (interval collapses).}
      The bounds meet, i.e.\ $h(X\mid C)=h(X)=h(X\mid C)+H(C)$, if and only if
      \begin{equation}
      \label{eq:meet_condition}
      H(C)=0\quad \text{(equivalently, }\exists c^\star\text{ with }\pi_{c^\star}=1\text{)}.
      \end{equation}
\end{enumerate}
\end{proposition}

\begin{proof}
\textbf{1) Lower bound tightness.}
From \eqref{eq:gap_lower_exact}, $h(X)=h(X\mid C)$ holds if and only if $I(X;C)=0$.
A standard information-theoretic equivalence is that $I(X;C)=0$ if and only if $X$ and $C$ are independent.
Independence means $p_{X\mid C}(x\mid c)=p_X(x)$ for every $c$ (almost everywhere in $x$).
Since $p_{X\mid C}(x\mid c)=f_c(x)$, this is exactly $f_c(x)=p_X(x)$ a.e.\ for all $c$.

\textbf{2) Upper bound tightness.}
From \eqref{eq:gap_upper_exact}, $h(X)=h(X\mid C)+H(C)$ holds if and only if $H(C\mid X)=0$.
For discrete $C$, $H(C\mid X)=0$ if and only if the conditional pmf $\mathbb{P}(C=\cdot\mid X=x)$ is a point mass
for almost every $x$ (i.e., $C$ is almost surely determined by $X$).
Equivalently, there exists a measurable decision rule $g$ such that $C=g(X)$ almost surely.

\textbf{3) Bounds meet.}
If the bounds meet, then both equalities in items 1) and 2) hold simultaneously, so $I(X;C)=0$ and $H(C\mid X)=0$.
But the identity $I(X;C)=H(C)-H(C\mid X)$ implies that if $H(C\mid X)=0$ then $I(X;C)=H(C)$.
Thus $I(X;C)=0$ and $H(C\mid X)=0$ together imply $H(C)=0$.
Conversely, if $H(C)=0$, then $C$ is deterministic (one component has probability one) and $h(X)=h(X\mid C)$ while
$h(X)=h(X\mid C)+H(C)$ holds trivially, so the interval collapses.
\end{proof}

\begin{remark}[Non-degenerate mixtures]
If $H(C)>0$ (at least two components have nonzero weight), then the interval in
Theorem~\ref{thm:label_sandwich} has positive width $H(C)$, so the two bounds cannot coincide exactly.
In practice, one bound becomes \emph{effectively tight} in the appropriate regime:
heavy overlap gives $I(X;C)\approx 0$ (lower bound nearly tight), while negligible overlap gives $H(C\mid X)\approx 0$
(upper bound nearly tight).
\end{remark}

\section{Universal Approximation}
\label{sec:approx}

The universal bounds in Theorem~\ref{thm:label_sandwich} are conceptually appealing:
they hold for \emph{any} mixture distribution and require no assumptions beyond the
existence of component entropies.  However, these bounds alone do not resolve the
practical problem of estimating $h(X)$, because the gap between
$h(X\mid C)$ and $h(X\mid C)+H(C)$ is precisely the mutual information
$I(X;C)$, which can vary continuously from $0$ to $H(C)$ as the degree of
component overlap changes.  In other words, the bounds are tight only at the
extremes—complete overlap or complete separation—but may be loose in the
intermediate regime where most practical systems operate.

Our goal in this section is therefore to construct a \emph{single closed-form expression} that tracks $h(X)$ across all overlap regimes, while retaining exactness in the two limits where the information-theoretic structure is fully understood.  The construction proceeds in two steps. First, we invoke Jensen’s inequality to obtain a computable lower bound in terms of pairwise component overlaps.  This step removes the logarithm from inside the integral and yields an expression that depends only on $\{z_{c,d}\}$. Second, we observe that this Jensen bound has a systematic bias that is independent of the mixture weights but depends only on the underlying component family.  By compensating for this bias with an analytically chosen offset, we obtain an approximation that is exact when all components coincide and asymptotically exact when components are well separated.  A final clipping step enforces consistency with the universal bounds of Theorem~\ref{thm:label_sandwich}.

 \subsection{Useful Lemmas}
\label{subsec:jensen_overlap}
Before presenting the approximate expressions of $h(X)$, we introduce lemmas that form the foundation of our closed form approximation framework. The first lemma provides a deterministic lower bound on the mixture entropy by moving the logarithm outside the intractable integral via Jensen’s inequality. The second lemma characterizes the R\'enyi-2 (collision) entropy of a single component density and the third lemma shows the offset-compensation mechanism to calibrate the Jensen bound using R\'enyi-2 entropy.

\begin{lemma}[Jensen/overlap lower bound]
\label{lem:jensen_lower}
For any mixture density $p_X(x)=\sum_{c=1}^K \pi_c f_c(x)$,
\begin{equation}
\label{eq:hJ_def}
h(X)\ \ge\ h_{\sf L}(X)
\triangleq
-\sum_{c=1}^K \pi_c\,
\logtwo\! \left(\sum_{d=1}^K \pi_d\,z_{c,d} \right),
\end{equation}
where $z_{c,d}=\int_{\R^n} f_c(x)f_d(x)\,dx$.
\end{lemma}

\begin{proof}
Start from the entropy definition and expand $p_X$ by conditioning on $C$:
\begin{align}
h(X)
&=-\int_{\R^n} p_X(x)\logtwo p_X(x)\,dx \nonumber\\
&=-\int_{\R^n} \Big(\sum_{c=1}^K \pi_c f_c(x)\Big)\logtwo p_X(x)\,dx \nonumber\\
&=-\sum_{c=1}^K \pi_c \int_{\R^n} f_c(x)\logtwo p_X(x)\,dx.
\label{eq:h_expand_by_components}
\end{align}
Fix an index $c$. Let $Y_c\sim f_c$. Then
\begin{equation}
\label{eq:component_term_as_expectation}
-\int_{\R^n} f_c(x)\logtwo p_X(x)\,dx
=\E\!\left[-\logtwo p_X(Y_c)\right].
\end{equation}
The function $\phi(t)=-\logtwo t$ is convex on $(0,\infty)$.
Therefore Jensen's inequality gives
\begin{equation}
\label{eq:jensen_step}
\E\!\left[-\logtwo p_X(Y_c)\right]
\ge -\logtwo\!\left(\E[p_X(Y_c)]\right).
\end{equation}
Compute $\E[p_X(Y_c)]$ explicitly:
\begin{align}
\E[p_X(Y_c)]
&=\int_{\R^n} f_c(x)\,p_X(x)\,dx \nonumber\\
&=\int_{\R^n} f_c(x)\,\left(\sum_{d=1}^K \pi_d f_d(x)\right)\,dx \nonumber\\
&=\sum_{d=1}^K \pi_d \int_{\R^n} f_c(x)f_d(x)\,dx \nonumber\\
&=\sum_{d=1}^K \pi_d z_{c,d}.
\label{eq:overlap_expectation}
\end{align}
Plugging \eqref{eq:overlap_expectation} into \eqref{eq:jensen_step} and then into \eqref{eq:h_expand_by_components},
\begin{align}
h(X)
&\ge -\sum_{c=1}^K \pi_c \logtwo\!\left(\sum_{d=1}^K \pi_d z_{c,d}\right)
= h_{\sf L}(X),
\end{align}
which proves \eqref{eq:hJ_def}.
\end{proof}

\begin{lemma}[R\'enyi--2 (collision) entropy of a density]
\label{lem:renyi2}
Let $f$ be a probability density on $\R^n$ such that $\int_{\R^n} f(x)^2\,dx<\infty$.
Define the collision entropy
\begin{equation}
\label{eq:renyi2_def}
h_2(f)\triangleq -\logtwo \int_{\R^n} f(x)^2\,dx.
\end{equation}
If $X\sim f$, then
\begin{equation}
\label{eq:renyi2_interp}
h_2(f)=-\logtwo \E[f(X)].
\end{equation}
Moreover, $h_2(f)\le h(f)$.
\end{lemma}

\begin{proof}
If $X\sim f$, then by definition of expectation,
\begin{align}
\E[f(X)] = \int_{\R^n} f(x)\,f(x)\,dx=\int_{\R^n} f(x)^2\,dx,
\end{align}
which gives \eqref{eq:renyi2_interp} and is equivalent to \eqref{eq:renyi2_def}.

To show $h_2(f)\le h(f)$, note that for $X\sim f$,
\begin{equation}
h(f)=-\E[\logtwo f(X)].
\end{equation}
Apply Jensen's inequality to the convex function $-\logtwo(\cdot)$:
\begin{align}
h(f)=-\E[\logtwo f(X)]
\ge -\logtwo \E[f(X)]
= h_2(f),
\end{align}
proving $h_2(f)\le h(f)$.
\end{proof}

\begin{lemma}[Offset compensation via the Shannon--collision gap]
\label{lem:offset_comp}
Let $p_X(x)=\sum_{c=1}^K \pi_c f_c(x)$.
Assume that for each component $f_c$, the gap
\begin{equation}
\label{eq:gap_def}
\Delta(f_c)\triangleq h(f_c)-h_2(f_c)
\end{equation}
is a known constant depending only on the component family (and dimension), not on location/scale parameters.
Define
\begin{equation}
\label{eq:Delta_c_def}
\Delta_c\triangleq \Delta(f_c),\qquad \bar\Delta\triangleq \sum_{c=1}^K \pi_c\Delta_c,
\end{equation}
and the offset-compensated Jensen functional
\begin{equation}
\label{eq:hJoff_def}
{\hat h}(X)\triangleq h_{\sf L}(X)+\bar\Delta.
\end{equation}
Then:
\begin{enumerate}[leftmargin=*, itemsep=2pt]
\item \textbf{Complete-overlap exactness.} If $f_1=\cdots=f_K=f$, then ${\hat h}(X)=h(f)=h(X)$.
\item \textbf{Zero-overlap exactness.} If $z_{c,d}=0$ for all $c\neq d$ (no cross-overlap), then
      ${\hat h}(X)=h(X\mid C)+H(C)$.
\item \textbf{Family constants.} For the families in this paper,
\begin{equation}
\label{eq:offset_constants}
\Delta_c=
\begin{cases}
\displaystyle \frac{n}{2}\log_2\!\Big(\frac{e}{2}\Big), & f_c\ \text{Gaussian on }\R^n,\\[0.9ex]
\displaystyle n\log_2\!\Big(\frac{e}{2}\Big), & f_c\ \text{factorized Laplacian on }\R^n,\\[0.9ex]
0, & f_c\ \text{uniform on a measurable set}.
\end{cases}
\end{equation}
\end{enumerate}
\end{lemma}

\begin{proof}
\textbf{1) Complete overlap.}
Assume $f_c\equiv f$ for all $c$. Then $p_X=f$ and every overlap is the same:
\begin{align}
z_{c,d}=\int_{\R^n} f(x)f(x)\,dx=\int_{\R^n} f(x)^2\,dx,\qquad \forall c,d.
\end{align}
Therefore for any $c$,
\begin{align}
\sum_{d=1}^K \pi_d z_{c,d}
&=\left(\sum_{d=1}^K \pi_d\right)\int_{\R^n} f(x)^2\,dx\nonumber\\
&=\int_{\R^n} f(x)^2\,dx.
\end{align}
Plug into the definition of $h_{\sf L}(X)$:
\begin{align}
h_{\sf L}(X)
&=-\sum_{c=1}^K \pi_c \logtwo\!\Big(\int_{\R^n} f(x)^2\,dx\Big)\nonumber\\
&=-\logtwo\!\int_{\R^n} f(x)^2\,dx
= h_2(f),
\end{align}
using $\sum_c\pi_c=1$. By assumption, $\Delta(f)=h(f)-h_2(f)$ is a constant for the family, and
$\bar\Delta=\sum_c\pi_c\Delta(f)=\Delta(f)$. Hence
\begin{align}
{\hat h}(X)
&=h_{\sf L}(X)+\bar\Delta \nonumber\\
&=h_2(f)+\big(h(f)-h_2(f)\big) \nonumber\\
&=h(f)=h(X).
\end{align}

\textbf{2) Zero overlap.}
Assume $z_{c,d}=0$ for all $c\neq d$. Then for each $c$,
\begin{align}
\sum_{d=1}^K \pi_d z_{c,d}=\pi_c z_{c,c}.
\end{align}
Substitute into $h_{\sf L}(X)$:
\begin{align}
h_{\sf L}(X)
&=-\sum_{c=1}^K \pi_c \logtwo(\pi_c z_{c,c}) \nonumber\\
&=-\sum_{c=1}^K \pi_c \logtwo \pi_c\ -\ \sum_{c=1}^K \pi_c \logtwo z_{c,c}.
\label{eq:hL_zero_overlap_expand}
\end{align}
The first term equals $H(C)$ by definition.
For the second term, note that $z_{c,c}=\int f_c(x)^2dx$, hence
\begin{align}
-\logtwo z_{c,c} = -\logtwo\!\int f_c(x)^2\,dx = h_2(f_c).
\end{align}
Therefore \eqref{eq:hL_zero_overlap_expand} becomes
\begin{align}
h_{\sf L}(X)=H(C)+\sum_{c=1}^K \pi_c\,h_2(f_c).
\end{align}
Add the offset:
\begin{align}
{\hat h}(X)
&=H(C)+\sum_{c=1}^K \pi_c\,h_2(f_c)+\sum_{c=1}^K \pi_c\Delta_c \nonumber\\
&=H(C)+\sum_{c=1}^K \pi_c\big(h_2(f_c)+h(f_c)-h_2(f_c)\big)\nonumber\\
&=H(C)+\sum_{c=1}^K \pi_c h(f_c)\nonumber\\
&=H(C)+h(X\mid C),
\end{align}
as claimed.

\textbf{3) Family constants.}
\emph{Gaussian:} for $f(x)=\mathcal{N}(\mu,\Sigma)$ on $\R^n$,
\begin{align}
h(f)&=\frac12\logtwo\big((2\pi e)^n|\Sigma|\big).
\end{align}
Also,
\begin{align}
\int_{\R^n} f(x)^2\,dx
&=\int \frac{1}{(2\pi)^{n}|\Sigma|}\exp\!\Big(-(x-\mu)^\top \Sigma^{-1}(x-\mu)\Big)\,dx \nonumber\\
&=\frac{1}{(2\pi)^{n}|\Sigma|}\cdot (\pi)^{n/2}|\Sigma|^{1/2}
=(4\pi)^{-n/2}|\Sigma|^{-1/2},
\end{align}
using the standard Gaussian integral $\int \exp(-y^\top A y)\,dy=\pi^{n/2}|A|^{-1/2}$ for $A\succ 0$
(with $A=\Sigma^{-1}$ after shifting).
Thus
\begin{align}
h_2(f)
&=-\logtwo\big((4\pi)^{-n/2}|\Sigma|^{-1/2}\big)
=\frac12\logtwo\big((4\pi)^n|\Sigma|\big).
\end{align}
Therefore
\begin{align}
\Delta_{\sf G}=h(f)-h_2(f)
=\frac12\logtwo\!\left(\frac{(2\pi e)^n}{(4\pi)^n}\right)
=\frac{n}{2}\logtwo\!\Big(\frac{e}{2}\Big).
\end{align}

\emph{Factorized Laplacian:} in one dimension, $f(x)=\frac{1}{2b}e^{-|x-\mu|/b}$.
Then
\begin{align}
h(f)&=\logtwo(2eb)\end{align}
and
\begin{align}
\int_{0}^{\infty} f(x)^2dx
&=\int_{0}^{\infty} \frac{1}{4b^2}e^{-2|x-\mu|/b}\,dx \nonumber\\
&=\frac{1}{4b^2} 2\int_{0}^{\infty}e^{-2t/b}\,dt \nonumber\\
&=\frac{1}{4b},
\end{align}
hence $h_2(f)=\logtwo(4b)$ and $\Delta=\logtwo(e/2)$.
For $n$ independent coordinates, entropies add, giving $\Delta_{\sf L}=n\logtwo(e/2)$.

\emph{Uniform on a set:} for $f(x)=\frac{1}{|\mathcal{A}|}\mathbf{1}\{x\in\mathcal{A}\}$,
$h(f)=\logtwo|\mathcal{A}|$ and $\int f^2=|\mathcal{A}|^{-1}$ so $h_2(f)=\logtwo|\mathcal{A}|$ and $\Delta_{\sf U}=0$.
This completes the proof.
\end{proof}

 \subsection{A Tight Approximation in Closed Form}
\label{subsec:thm2}
Using the lemmas above, we derive a practically tight closed-form approximation valid across mixture regimes.

\begin{theorem}[Approximation in closed form]
\label{thm:tight_approx}
Define the component-wise offsets
\begin{equation}
\label{eq:Delta_c_def}
\Delta_c=
\begin{cases}
\displaystyle \frac{n}{2}\log_2\!\Big(\frac{e}{2}\Big), & f_c\ \text{Gaussian},\\[1.0ex]
\displaystyle n\log_2\!\Big(\frac{e}{2}\Big), & f_c\ \text{factorized Laplacian},\\[1.0ex]
0, & f_c\ \text{uniform},
\end{cases}
\end{equation}
and their mixture average
\[
\bar{\Delta}\triangleq \sum_{c=1}^K \pi_c \Delta_c.
\]
The proposed approximation for $h(X)$
\begin{align}
   {\hat h}(X)&\triangleq h_{\sf L}(X)+\bar{\Delta},\nonumber\\
   &=-\sum_{c=1}^K \pi_c\, \log_2\! \left(\sum_{d=1}^K \pi_d\,z_{c,d} \right)+\sum_{c=1}^K \pi_c \Delta_c\nonumber\\
   &=-\sum_{c=1}^K \pi_c\, \left[\log_2\! \left(\sum_{d=1}^K \pi_d\,z_{c,d} \right)+  \Delta_c\right].
\end{align}
and the clipped estimate 
\begin{equation}
\label{eq:h_clip}
 {\hat h}_{\sf cl}(X)\triangleq
\min\{ \max\{{\hat h}(X), h(X\mid C)\}, h(X\mid C)+H(C)\}.
\end{equation}
\end{theorem}

\begin{proof}
The proof follows directly from the preceding lemmas.
Lemma~\ref{lem:jensen_lower} establishes that $h_{\sf L}(X)$ is a valid lower bound on $h(X)$.
Lemma~\ref{lem:offset_comp} characterizes the constant gap between the Jensen functional
and the Shannon entropy in the complete-overlap regime, enabling the offset
$\bar{\Delta}$ to remove this bias.
Exactness in the complete-overlap case follows immediately from this calibration,
while exactness in the vanishing-overlap regime follows from the fact that
$h(X)\to h(X\mid C)+H(C)$ as $I(X;C)\to H(C)$.
Finally, clipping enforces the universal bounds of
Theorem~\ref{thm:label_sandwich}, guaranteeing consistency for all mixtures.
\end{proof}

 \subsection{Implications of Theorem~\ref{thm:tight_approx}}
\label{subsec:imp_thm2}

The starting point is the decomposition $h(X)=h(X\mid C)+I(X;C)$. For most mixture models,
$h(X\mid C)$ is easy—it's just the weighted average of component entropies—while the hard part is the
 missing term   $I(X;C)$, which is entirely about how much the components overlap.
Theorem~\ref{thm:tight_approx} says we can replace this implicit, posterior-based quantity by a
\emph{deterministic} proxy computable from pairwise overlaps: in effect,
\begin{align}
   \hat h(X)=h(X\mid C)+\widehat I(X;C),
\end{align}
where $\widehat I(X;C)\triangleq \hat h(X)-h(X\mid C)$ is driven by the overlap geometry through the Jensen functional.
When the cross-overlaps $z_{c,d}$ are large, the mixture behaves like a single cloud and the proxy
predicts little information about the label; when cross-overlaps are small, the proxy predicts that
the label is nearly decodable. In this sense, the overlap matrix plays the role of a \textit{soft confusion matrix} for the components, and $\hat h(X)$ is the corresponding entropy
bookkeeping in closed form.

Just as importantly, Theorem~\ref{thm:tight_approx} explains why the construction is stable.
The offset $\bar\Delta$ is not a free parameter: it is the intrinsic gap between Shannon entropy and
R\'enyi-$2$ entropy at the \emph{component} level, $\Delta(f)=h(f)-h_2(f)$, averaged across components.
This calibrates the overlap-based proxy so that it lands correctly at the two information-theoretic
endpoints: complete overlap ($I(X;C)=0$ so $h(X)=h(X\mid C)$) and vanishing overlap
($I(X;C)\approx H(C)$ so $h(X)\approx h(X\mid C)+H(C)$). The final clipping step then enforces the
universal label-sandwich bounds for every mixture, regardless of dimension or parameter choice:
\[
h(X\mid C)\ \le\ \hat h(X)\ \le\ h(X\mid C)+H(C),
\]
so the approximation cannot \textit{cheat} basic information identities in the intermediate regime. Operationally, the theorem reduces the problem of approximating $h(X)$ to one question: can we compute overlaps $z_{c,d}$ (exactly or efficiently)? If yes, then we immediately get a
closed-form surrogate for both $h(X)$ and the mutual information $I(X;C)$ via $\widehat I(X;C)$.

\section{Examples for Time-Honored Mixture Families}
\label{sec:examples}

This section instantiates Theorem~\ref{thm:label_sandwich} and Theorem~\ref{thm:tight_approx} for several time-honored mixture families that frequently appear across information theory, signal processing, and statistical learning. For each family we provide (i) the closed-form conditional entropy $h(X\mid C)$, (ii) the overlap integrals $z_{c,d}=\int f_c(x) f_d(x){\rm d}x$, and (iii) the corresponding offset $\bar\Delta$. Once these ingredients are available, the Jensen proxy $h_{\sf L}(X)$ in \eqref{eq:hJ_def} and the clipped approximation $\hat h(X)$ in \eqref{eq:h_clip} follow directly.

In applications, mixture entropy is rarely needed in isolation---it typically enters a larger objective: capacity/achievable-rate expressions for discrete-input continuous-output channels, regularization terms in variational learning, or diagnostics for cluster separability. The overlap matrix $\{z_{c,d}\}$ is especially useful because it acts as a \emph{soft confusion matrix}: small cross-overlaps indicate that the label $C$ is (nearly) decodable from $X$ so $h(X)\approx h(X\mid C)+H(C)$, whereas large cross-overlaps indicate ambiguity so $h(X)\approx h(X\mid C)$.

\begin{table*}[t]
\centering
\small
\renewcommand{\arraystretch}{1.20}
\setlength{\tabcolsep}{6.2pt}
\begin{tabular}{p{0.17\textwidth} p{0.30\textwidth} p{0.2\textwidth} p{0.2\textwidth}}
\toprule
\textbf{Family} & $\boldsymbol{h(X\mid C)}$ & \textbf{Overlaps $\boldsymbol{z_{c,d}}$} & $\boldsymbol{\bar\Delta}$ \\
\midrule
Gaussian
& $\sum_{c}\pi_c\frac12\logtwo\!\big((2\pi e)^n|\Sigma_c|\big)$
& $z_{c,d}^{\sf GM}=\varphi(\mu_c;\mu_d,\Sigma_c+\Sigma_d)$
& $\frac{n}{2}\logtwo(e/2)$ \\[0.6ex]

Factorized Laplacian
& $\sum_c \pi_c\sum_{i=1}^n \logtwo(2e\,b_{c,i})$
& $z_{c,d}^{\sf LM}=\prod_i z^{\mathrm{LL}}_{c,d}(i)$ (closed form)
& $n\logtwo(e/2)$ \\[0.6ex]

Uniform on sets
& $\sum_c \pi_c\logtwo|\mathcal{A}_c|$
& $z_{c,d}^{\sf UM}=\frac{|\mathcal{A}_c\cap \mathcal{A}_d|}{|\mathcal{A}_c||\mathcal{A}_d|}$
& $0$ \\[0.6ex]

Gaussian--Uniform
& mixture-weighted sum of Gaussian and uniform entropies
& $z_{c,d}^{\sf GUM}=\frac{1}{|\mathcal{A}_c|}\,\mathbb{P}(Z_d\in\mathcal{A}_c)$
& $\sum_{c\in\mathcal{G}}\pi_c\frac{n}{2}\logtwo(e/2)$ \\[0.6ex]

Gaussian--Laplacian (diag.\ $\Sigma$)
& mixture-weighted sum of Gaussian and Laplace entropies
& $z_{c,d}^{\sf GLM}$ in \eqref{eq:gauss_lap_overlap_1d_ex}.
& $\sum_{c\in\mathcal{G}}\pi_c\frac{n}{2}\logtwo(e/2)+\sum_{c\in\mathcal{L}}\pi_c\,n\logtwo(e/2)$ \\[0.6ex]

Laplacian--Uniform
& mixture-weighted sum of Laplace and uniform entropies
& $z_{c,d}^{\sf LUM}=\frac{1}{|\mathcal{A}_d|}\mathbb{P}(Y_c\in\mathcal{A}_d).$
& $\sum_{c\in\mathcal{L}}\pi_c\,n\logtwo(e/2)$ \\
\bottomrule
\end{tabular}
\vspace{0.1cm}
\caption{Closed-form ingredients for evaluating $h_{\sf L}(X)$ and the approximation $\hat h(X)$.}
\label{tab:mixture_families_summary}
\end{table*}

\subsection{Example 1: Gaussian mixture}
\label{subsec:ex_gaussian}
 
Gaussian mixtures are the workhorse model in model-based clustering, discriminant analysis, and latent-variable
learning---typically fit by EM-type procedures
\cite{BanfieldRaftery1993ModelBased,McLachlanPeel2000Mixtures,FraleyRaftery2002ModelBasedClustering,DempsterLairdRubin1977EM}.
From an information-theoretic viewpoint, Gaussian mixtures also arise as channel output densities when a discrete
constellation is transmitted over continuous-output channels, a classical setting behind finite-input capacity
bounds and approximations \cite{Smith1971AmplitudeConstraint,OzarowWyner1990FiniteInput,namyoon,sumin}.
Because of these roles, computable entropy surrogates for Gaussian mixtures have been studied extensively
(e.g., Jensen-type bounds in \cite{HuberBaileyDurrantWhyteHanebeck2008} and tight bounds in \cite{MoshksarKhandani2016EntropyGM}).

Let $f_c(x)=\varphi(x;\mu_c,\Sigma_c)$ with $\Sigma_c\succ 0$, where
\begin{equation}
\label{eq:gauss_pdf}
\varphi(x;\mu,\Sigma)
\triangleq
\frac{1}{(2\pi)^{n/2}|\Sigma|^{1/2}}
\exp\!\Big(-\tfrac12(x-\mu)^\top\Sigma^{-1}(x-\mu)\Big).
\end{equation}
Then
\begin{equation}
\label{eq:hXgC_gauss_ex}
h(X\mid C)=\sum_{c=1}^K \pi_c\,\frac12\logtwo\!\big((2\pi e)^n|\Sigma_c|\big).
\end{equation}
The overlaps admit the closed form
\begin{align}
\label{eq:zcd_gg_ex}
z_{c,d}
&=\int_{\R^n} \varphi(x;\mu_c,\Sigma_c)\varphi(x;\mu_d,\Sigma_d)\,dx \nonumber\\
&=\varphi(\mu_c;\mu_d,\Sigma_c+\Sigma_d) \nonumber\\
&=\frac{\exp\!\Big(-\frac12(\mu_c-\mu_d)^\top(\Sigma_c+\Sigma_d)^{-1}(\mu_c-\mu_d)\Big)}
{(2\pi)^{n/2}\,|\Sigma_c+\Sigma_d|^{1/2}}.
\end{align}
Substituting into \eqref{eq:hJ_def} gives
\begin{align}
\label{eq:hJ_gauss_explicit}
h_{\sf L}^{\sf GM}(X)
=
-\sum_{c=1}^K \pi_c\logtwo\!\Big(\sum_{d=1}^K \pi_d\,\varphi(\mu_c;\mu_d,\Sigma_c+\Sigma_d)\Big).
\end{align}
For a pure Gaussian mixture, $\bar\Delta=\frac{n}{2}\logtwo(e/2)$, and $\hat h(X)$ follows from \eqref{eq:h_clip}.

 The overlap $z_{c,d}$ is a Gaussian density evaluated at $\mu_c$ under mean $\mu_d$ and covariance $\Sigma_c+\Sigma_d$. It decays exponentially in Mahalanobis separation, so increasing mean separation increases the effective decodability of the label (larger $I(X;C)$), pushing $h(X)$ from $h(X\mid C)$ toward $h(X\mid C)+H(C)$. In clustering terms, $z_{c,d}$ quantifies \emph{confusability} between clusters; in discrete-input channels, it quantifies how likely two inputs induce overlapping output clouds, which directly impacts mutual information and hence rate.

\subsection{Example 2: Factorized Laplacian mixture}
\label{subsec:ex_laplace}

 Laplacian components (and mixtures thereof) arise in heavy-tailed modeling, robust estimation, and sparsity-driven signal models; in statistics, mixtures of Laplace-like components provide flexible tail and skewness modeling (e.g., shifted asymmetric Laplace mixtures) \cite{FranczakBrowneMcNicholas2014SALMixtures}. In early applied mixture modeling, Laplace--normal (Gaussian--Laplacian) hybrids were used to capture impulsive departures from nominal behavior \cite{JonesMcLachlan1990LaplaceNormal}, a theme that persists in robust modeling.

 Let
\[
f_c(x)=\prod_{i=1}^n \frac{1}{2b_{c,i}}
\exp\!\Big(-\frac{|x_i-\mu_{c,i}|}{b_{c,i}}\Big),\qquad b_{c,i}>0.
\]
Then
\begin{equation}
\label{eq:hXgC_lap_ex}
h(X\mid C)=\sum_{c=1}^K \pi_c \sum_{i=1}^n \logtwo(2e\,b_{c,i}).
\end{equation}
Overlaps factorize as $z_{c,d}=\prod_{i=1}^n z^{\mathrm{LM}}_{c,d}(i)$ where, with
$\delta_i=|\mu_{c,i}-\mu_{d,i}|$, $s_1=b_{c,i}$, $s_2=b_{d,i}$,
\begin{equation}
\label{eq:lap_overlap_1d_ex}
z^{\mathrm{LM}}_{c,d}(i)
=
\begin{cases}
\displaystyle
\frac{e^{-\delta_i/s}}{4s}\Big(1+\frac{\delta_i}{s}\Big), & s_1=s_2=s,\\[1.2ex]
\displaystyle
\frac{s_1e^{-\delta_i/s_1}-s_2e^{-\delta_i/s_2}}{2(s_1^2-s_2^2)}, & s_1\neq s_2.
\end{cases}
\end{equation}
Thus
\begin{equation}
\label{eq:hJ_lap_explicit}
h_{\sf L}^{\sf LM}(X)
=
-\sum_{c=1}^K \pi_c\logtwo\!\Big(\sum_{d=1}^K \pi_d\,
\prod_{i=1}^n z^{\mathrm{LM}}_{c,d}(i)\Big).
\end{equation}
For a pure Laplacian mixture, $\bar\Delta=n\logtwo(e/2)$.

The product structure makes overlap computation scale as $O(K^2 n)$ rather than $O(K^2)$ high-dimensional integrals,
which is critical in large-$n$ robust models.
Each coordinate contributes a multiplicative penalty as $|\mu_{c,i}-\mu_{d,i}|$ grows, so even moderate per-coordinate
separation can rapidly shrink cross-overlaps in high dimension.
Compared to Gaussians (variance matched), heavier tails generally sustain overlap longer in the intermediate regime,
which explains why Laplacian-mixture entropy can approach the upper bound more gradually as separation increases.

\subsection{Example 3: Uniform mixture on measurable sets}
\label{subsec:ex_uniform_sets}

Uniform components are natural for \emph{hard support uncertainty}: bounded perturbations, feasibility sets,
quantization cells, and saturation/clipping constraints.
Uniform mixtures have been studied from the standpoint of statistical estimation and identifiability
\cite{CraigmileTitterington1997UniformMixtures}, including tractable special cases such as two-component
uniform mixtures \cite{HusseinLiu2009TwoUniforms}.
In engineering, mixtures of uniform laws also serve as coarse priors or as bounded uncertainty modes in robust inference.

Let $f_c$ be uniform on $\mathcal{A}_c\subset\R^n$ with finite volume $|\mathcal{A}_c|$:
\[
f_c(x)=\frac{1}{|\mathcal{A}_c|}\mathbf{1}\{x\in\mathcal{A}_c\}.
\]
Then
\begin{align}
\label{eq:hXgC_unif_ex}
h(X\mid C)&=\sum_{c=1}^K \pi_c \logtwo|\mathcal{A}_c|,\\
\label{eq:unif_overlap_ex}
\end{align}
and
\begin{align}
z_{c,d}&=\frac{|\mathcal{A}_c\cap \mathcal{A}_d|}{|\mathcal{A}_c|\,|\mathcal{A}_d|}.
\end{align}
Hence
\begin{equation}
\label{eq:hJ_unif_explicit}
h_{\sf L}^{\sf UM}(X)
=
-\sum_{c=1}^K \pi_c\logtwo\!\left(
\sum_{d=1}^K \pi_d\,
\frac{|\mathcal{A}_c\cap \mathcal{A}_d|}{|\mathcal{A}_c|\,|\mathcal{A}_d|}
\right).
\end{equation}
For uniform components, $\bar\Delta=0$.

The approximation becomes purely geometric: only volumes and intersection volumes appear.
This connects mixture entropy to set packing/covering intuition: when $\mathcal{A}_c$ are nearly disjoint,
cross-overlaps vanish, the label becomes decodable, and the entropy rises to $h(X\mid C)+H(C)$ quickly.
This geometry-centric view is especially useful when supports have analytic intersection formulas
(e.g., axis-aligned boxes, certain polytopes), or when intersection queries can be computed efficiently.

\subsection{Example 4: Gaussian--uniform mixture}
\label{subsec:ex_gauss_unif}

Gaussian--uniform hybrids are a classical robust clustering/outlier model:
a Gaussian mixture captures nominal clusters while a uniform component absorbs background noise or gross outliers
\cite{CorettoHennig2011GaussianUniform}.
Recent work further analyzes separation constraints and identifiability properties in such models
\cite{Coretto2022GaussianUniformNoiseSeparation}.
In communications and detection problems, this hybrid also captures Gaussian noise plus bounded interference
or inlier/outlier sensing.

Partition labels into $\mathcal{G}$ (Gaussian) and $\mathcal{U}$ (uniform on sets).
Then
\begin{align}
\label{eq:hXgC_gauss_unif_ex}
h(X\mid C)
&=
\sum_{c\in\mathcal{U}} \pi_c \logtwo|\mathcal{A}_c|
+\sum_{c\in\mathcal{G}} \pi_c\,\frac12\logtwo\!\big((2\pi e)^n|\Sigma_c|\big).
\end{align}
Gaussian-uniform overlaps are
\begin{align}
\label{eq:unif_gauss_overlap_ex}
z_{c,d}^{\sf GUM}
&=\frac{1}{|\mathcal{A}_c|}\int_{\mathcal{A}_c}\varphi(x;\mu_d,\Sigma_d)\,dx
=\frac{1}{|\mathcal{A}_c|}\,\mathbb{P}(Z_d\in\mathcal{A}_c),
\end{align}
where $Z_d\sim\mathcal{N}(\mu_d,\Sigma_d)$. The offset averages only Gaussian components:
\[
\bar\Delta=\sum_{c\in\mathcal{G}} \pi_c \frac{n}{2}\logtwo(e/2).
\]

The mixed overlap in \eqref{eq:unif_gauss_overlap_ex} is Gaussian mass per unit volume of $\mathcal{A}_c$:
it quantifies how much of a Gaussian cloud falls inside a bounded-support mode.
In robust clustering terms \cite{CorettoHennig2011GaussianUniform}, small Uniform-Gaussian overlaps mean the uniform component
captures regions rarely visited by the Gaussians (easy outlier separation), which increases $I(X;C)$ and pushes
entropy upward toward the upper sandwich bound.

\subsection{Example 5: Gaussian--Laplacian mixture}
\label{subsec:ex_gauss_lap}

Gaussian--Laplacian hybrids model dense nominal with sparse/impulsive behavior and appear in applied mixture modeling
(e.g., Laplace--normal mixtures \cite{JonesMcLachlan1990LaplaceNormal}) and in robust system identification and learning
\cite{ShenoyGorinevsky2014GaussianLaplacian}.
From an information viewpoint, these hybrids are important because \emph{shape} differences (light vs.\ heavy tails)
can make labels partially decodable even when means coincide, especially in high dimension.

Partition labels into $\mathcal{G}$ (Gaussian) and $\mathcal{L}$ (factorized Laplacian). Then
\begin{align}
\label{eq:hXgC_gauss_lap_ex}
h(X\mid C)
&=
\sum_{c\in\mathcal{G}} \pi_c\,\frac12\logtwo\!\big((2\pi e)^n|\Sigma_c|\big)\nonumber\\
&+\sum_{c\in\mathcal{L}} \pi_c \sum_{i=1}^n \logtwo(2e\,b_{c,i}).
\end{align}
To retain closed form Gaussian-Laplacian overlaps, assume diagonal Gaussian covariance
$\Sigma_c=\diag(\sigma_{c,1}^2,\dots,\sigma_{c,n}^2)$ for $c\in\mathcal{G}$.
Then Gaussian-Laplacian overlaps factorize $z_{c,d}=\prod_{i=1}^n z^{\mathrm{GL}}_{c,d}(i)$ with
\begin{align}
\label{eq:gauss_lap_overlap_1d_ex}
z^{\mathrm{GLM}}_{c,d}(i)
&\triangleq
\int_{\mathbb R}\mathcal{N}(x;\mu_{c,i},\sigma_{c,i}^2)\,
\frac{1}{2b_{d,i}}e^{-|x-\mu_{d,i}|/b_{d,i}}\,dx \nonumber\\
&=
\frac{e^{\frac{\sigma_{c,i}^2}{2b_{d,i}^2}}}{2b_{d,i}}
\left[
e^{\frac{\mu_{c,i}-\mu_{d,i}}{b_{d,i}}}
\Phi\!\left(\frac{\mu_{d,i}-\mu_{c,i}-\frac{\sigma_{c,i}^2}{b_{d,i}}}{\sigma_{c,i}}\right) +
e^{-\frac{\mu_{d,i}\mu_{c,i}}{b_{d,i}}}
\Phi\!\left(\frac{\mu_{c,i}-\mu_{d,i}-\frac{\sigma_{c,i}^2}{b_{d,i}}}{\sigma_{c,i}}\right)
\right],
\end{align}
with $\Phi(\cdot)$ the standard normal cumulative distribution function (CDF). The offset averages Gaussian and Laplacian components:
\[
\bar\Delta
=
\sum_{c\in\mathcal{G}} \pi_c \frac{n}{2}\logtwo\!\Big(\frac{e}{2}\Big)
+\sum_{c\in\mathcal{L}} \pi_c\,n\logtwo\!\Big(\frac{e}{2}\Big).
\]

The diagonal assumption reduces an $n$-dimensional Gaussian-Laplacian overlap integral to $n$ scalar terms, yielding an $O(K^2 n)$ computation that scales to high-dimensional feature vectors. Operationally, the Gaussian-Laplacian overlaps quantify the soft match between a Gaussian bump and a Laplacian spike: if heavy-tailed components intrude substantially into nominal Gaussian regions, cross-overlaps increase, labels become harder to decode, and entropy remains closer to $h(X\mid C)$. This connects directly to robust modeling practices in \cite{ShenoyGorinevsky2014GaussianLaplacian}.

\subsection{Example 6: Laplacian--uniform mixture}
\label{subsec:ex_lap_unif}

Laplacian--uniform hybrids combine heavy-tailed variability (Laplacian) with bounded-support uncertainty (uniform), capturing scenarios where one mode is sparse/impulsive while another mode represents hard constraints or bounded background uncertainty.
This is a natural extension of (i) Laplacian mixture modeling used for tail/skewness flexibility
\cite{FranczakBrowneMcNicholas2014SALMixtures} and (ii) uniform mixture modeling used for bounded-support uncertainty \cite{CraigmileTitterington1997UniformMixtures,HusseinLiu2009TwoUniforms}.

Partition labels into $\mathcal{L}$ (factorized Laplacian) and $\mathcal{U}$ (uniform on sets).
For $c\in\mathcal{L}$ and $d\in\mathcal{U}$,
\begin{align}
z_{c,d}^{\sf LUM}
=\frac{1}{|\mathcal{A}_d|}\int_{\mathcal{A}_d} f_c(x)\,dx
=\frac{1}{|\mathcal{A}_d|}\mathbb{P}(Y_c\in\mathcal{A}_d),
\end{align}
where $Y_c\sim f_c.$ For the common case of axis-aligned boxes $\mathcal{A}_d=\mu_d+[-a_d,a_d]^n$ and factorized Laplace,
$\mathbb{P}(Y_c\in\mathcal{A}_d)$ factorizes into one-dimensional CDF differences:
\begin{align}
\mathbb{P}(Y_c\in\mathcal{A}_d)
=\prod_{i=1}^n
\left[F_{\mathrm{Lap}}(\mu_{d,i}+a_{d,i};\mu_{c,i},b_{c,i})  -F_{\mathrm{Lap}}(\mu_{d,i}-a_{d,i};\mu_{c,i},b_{c,i})\right],
\end{align}
where $F_{\mathrm{Lap}}(\cdot;\mu,b)$ is the Laplace CDF.
The offset averages only Laplacian components:
\[
\bar\Delta=\sum_{c\in\mathcal{L}}\pi_c\,n\logtwo(e/2).
\]

The Laplacian-Uniform overlap is the Laplacian probability mass captured by a bounded region, normalized by volume. It therefore measures how frequently impulsive/heavy-tailed realizations fall inside a hard feasibility set. When Laplacian-Uniform overlaps are small, the bounded-support components occupy regions rarely generated by Laplacian components, which increases label identifiability and pushes entropy upward. When Laplacian-Uniform overlaps are large, bounded regions act like typical sets' for Laplacian components, reducing identifiability
and keeping entropy near $h(X\mid C)$.


\section{Numerical Results}
\label{sec:numerics}

This section connects the theory in Theorem~\ref{thm:label_sandwich} and Theorem~\ref{thm:tight_approx} to
numerical evidence. Our simulations therefore sweep a single knob (mean separation) that continuously changes the overlap geometry,
and we track how $h(X)$ moves between the two universal endpoints
\begin{align}
h(X\mid C)\ \le\ h(X)\ \le\ h(X\mid C)+H(C),
\label{eq:num_sandwich}
\end{align}
while comparing to the proposed clipped approximation $\hat h(X)$.

\subsection{Simulation Setup}
\label{subsec:num_model}

Unless otherwise stated, we use equal weights
\begin{align}
p_X(x)=\frac{1}{K}\sum_{c=1}^K f_c(x), 
\label{eq:num_equal_weight}
\end{align}
with $H(C)=\logtwo K$ so the vertical separation between the label-sandwich bounds is fixed and interpretable as exactly $\logtwo K$ bits of label uncertainty.

To isolate the effect of overlap (not a changing geometry of directions), we draw $K$ random directions
$\{u_c\}_{c=1}^K$ once, normalize each to $\|u_c\|_2=1$, and set
\begin{align}
\mu_c(S)=S\,u_c,\qquad c=1,\dots,K,
\label{eq:num_means}
\end{align}
where $S\ge 0$ is the mean-separation parameter.
Thus:
\begin{itemize}[leftmargin=*, itemsep=2pt]
\item $S=0$ enforces \emph{maximal mean overlap} ($\mu_1=\cdots=\mu_K$), so any label information must come only
      from differences in \emph{shape} (e.g., Gaussian vs.\ Laplacian vs.\ uniform support).
\item Increasing $S$ reduces cross-overlaps $z_{c,d}$ and increases label decodability, so $I(X;C)$ increases and
      $h(X)$ rises from the lower bound toward the upper bound.
\end{itemize}

We sweep $S$ over a uniform grid and for each $S$ draw $N_s$ independent and identically distributed (i.i.d.) samples. The grid endpoints are chosen so that the curves visibly approach a plateau at high separation (near the upper bound). All reported Monte Carlo curves use the same fixed random directions $\{u_c\}$ across the sweep to make the dependence on $S$ comparable and monotone.

\subsection{Variance-Matched Component Families}
\label{subsec:num_families}

\begin{figure}[t]
\centering
\includegraphics[width=0.7\linewidth]{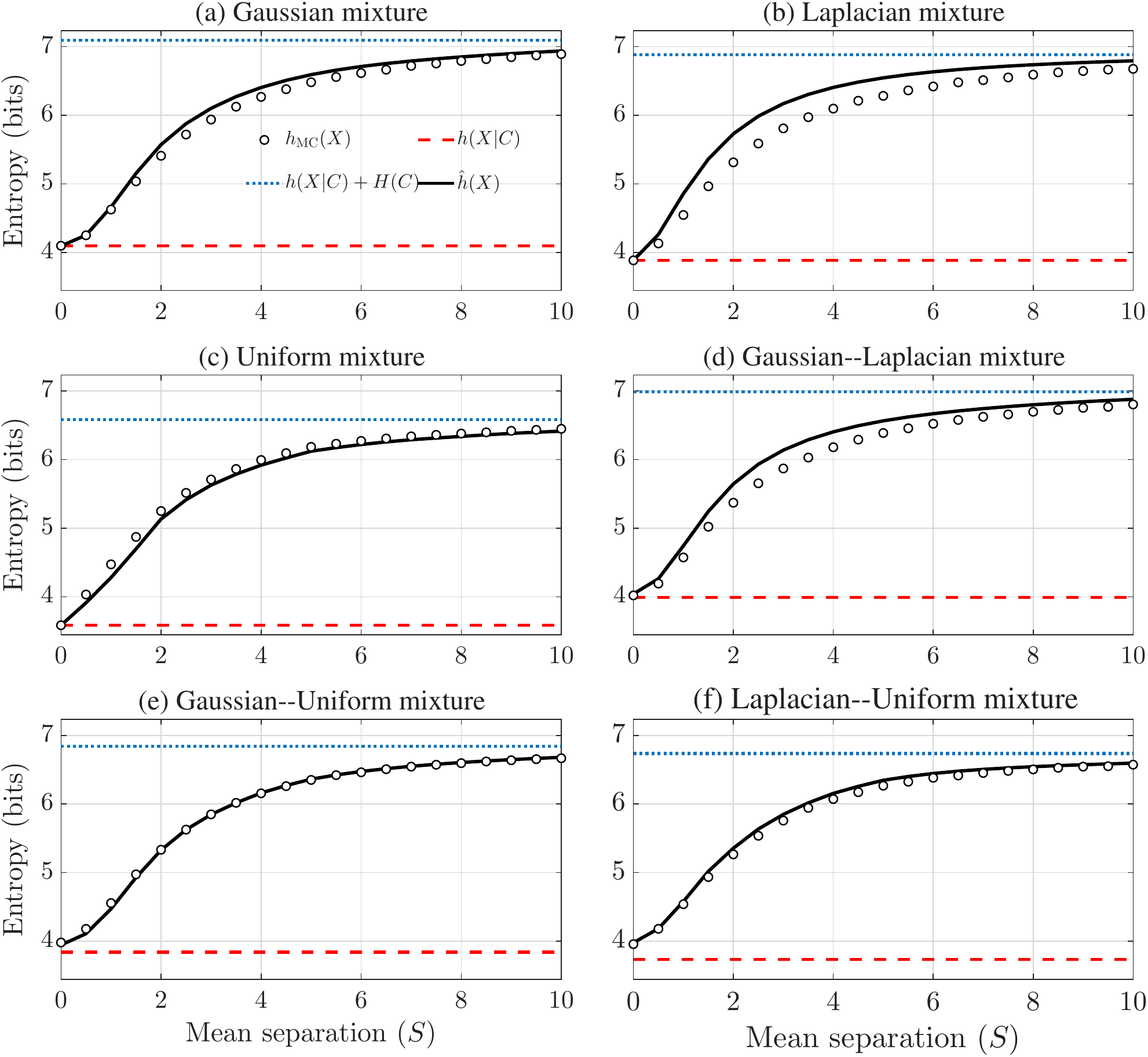}
\caption{Comparison of $h(X)$ for various mixture distributions when $(n,K)=(2,8)$: MC entropy, sandwich bounds, and clipped approximation versus mean separation.}
\label{fig:rep_case1_n2K8}
\end{figure}

In this section we test six mixture compositions: (a) Gaussian, (b) Laplacian, (c) Uniform, (d) Gaussian--Laplacian,
(e) Gaussian--Uniform, and (f) Laplacian--Uniform. All families are scaled to have the same per-dimension variance, to isolate the effect of \emph{shape}
(tails vs.\ compact support).  Specifically:
\begin{itemize}
    \item \text{Gaussian mixture:} \[p_X^{\sf GM}(x)=\frac{1}{K}\sum_{i=1}^K  \mathcal{N}(\mu_{i},\sigma^2)\] with $ \sigma^2=1$;
    \item \text{Laplacian mixture:} \[p_X^{\sf LM}(x)=\frac{1}{K}\sum_{i=1}^K  \mathrm{Lap}(\mu_{i},b)\] with $2b^2=\sigma^2 \Rightarrow b=\tfrac{1}{\sqrt{2}}$;
    \item  \text{Uniform mixture:} 
    \[p_X^{\sf UM}(x)=\frac{1}{K}\sum_{i=1}^K \mathrm{Unif}[\mu_{i}-a,\mu_{i}+a] \] with $\frac{a^2}{3}=\sigma^2 \Rightarrow a=\sqrt{3};$
     
    \item Gaussian-Laplace mixture: 
     \[p_X^{\sf GLM}(x)=\frac{2}{K}\sum_{i=1}^{K/2}  \mathcal{N}(\mu_{i},1) + \frac{2}{K}\sum_{i=1}^{K/2} \mathrm{Lap}\left(\mu_{i},\frac{1}{\sqrt{2}}\right); \]

        \item Gaussian-Uniform mixture: 
          \begin{align}
          p_X^{\sf GUM}(x)&=\frac{2}{K}\sum_{i=1}^{K/2}  \mathcal{N}(\mu_{i},1)  + \frac{2}{K}\sum_{i=1}^{K/2} \mathrm{Unif}[\mu_{i}-\sqrt{3},\mu_{i}+\sqrt{3}];
          \end{align}

     \item Laplacian-Uniform mixture: 
     \begin{align}
         p_X^{\sf LUM}(x)&=\frac{2}{K}\sum_{i=1}^{K/2}\mathrm{Lap}\left(\mu_{i},\frac{1}{\sqrt{2}}\right)  + \frac{2}{K}\sum_{i=1}^{K/2} \mathrm{Unif}[\mu_{i}-\sqrt{3},\mu_{i}+\sqrt{3}].\nonumber
     \end{align}
\end{itemize}
These hybrids are important because they expose a phenomenon invisible in same-family mixtures:
even at $S=0$ (identical means), labels can remain partially decodable due to \emph{shape mismatch}
(e.g., heavy tails vs.\ light tails, or bounded support vs.\ unbounded support).

\subsection{Benchmarks}
\label{subsec:num_metrics}

For each separation $S$, we compute four curves:
\begin{itemize}

\item Monte Carlo (MC) reference estimate: We draw $N_s$ i.i.d.\ samples $\{X^{(j)}\}_{j=1}^{N_s}\sim p_X$ and compute
\begin{align}
h_{\mathrm{MC}}(X)
\triangleq
-\frac{1}{N_s}\sum_{j=1}^{N_s}\logtwo p_X\!\big(X^{(j)}\big),
\label{eq:num_mc}
\end{align}
where $p_X(X^{(j)})=\sum_{c=1}^K \pi_c f_c(X^{(j)})$ is computed by a log-sum-exp of component log-densities.

    \item Label-sandwich bounds (Theorem~\ref{thm:label_sandwich}):  The lower bound is the closed-form conditional entropy
\begin{align}
h(X\mid C)=\sum_{c=1}^K \pi_c h(f_c),
\label{eq:num_hxgc}
\end{align}
which is constant in $S$ in our sweep because only the means change.
The upper bound is
\begin{align}
h(X\mid C)+H(C)=\sum_{c=1}^K \pi_c h(f_c)+\logtwo K.
\label{eq:num_upper}
\end{align}
These two horizontal reference lines define the only information-theoretically admissible interval for $h(X)$.

\item  Proposed approximation (Theorem~\ref{thm:tight_approx}): We compute the approximation in closed form as
\begin{align}
  {\hat h}_{\sf cl}(X) =
\min\{ \max\{{\hat h}(X), h(X\mid C)\}, h(X\mid C)+H(C)\}.
\label{eq:num_hhat}
\end{align}
All overlaps $z_{c,d}$ in computing $h(X\mid C)$ are evaluated in closed form for the considered families (Section~\ref{sec:examples}).
In computation, we evaluate $\log z_{c,d}$ and use log-sum-exp for numerical stability.
\end{itemize}


\begin{figure}[t]
\centering
\includegraphics[width=0.7\linewidth]{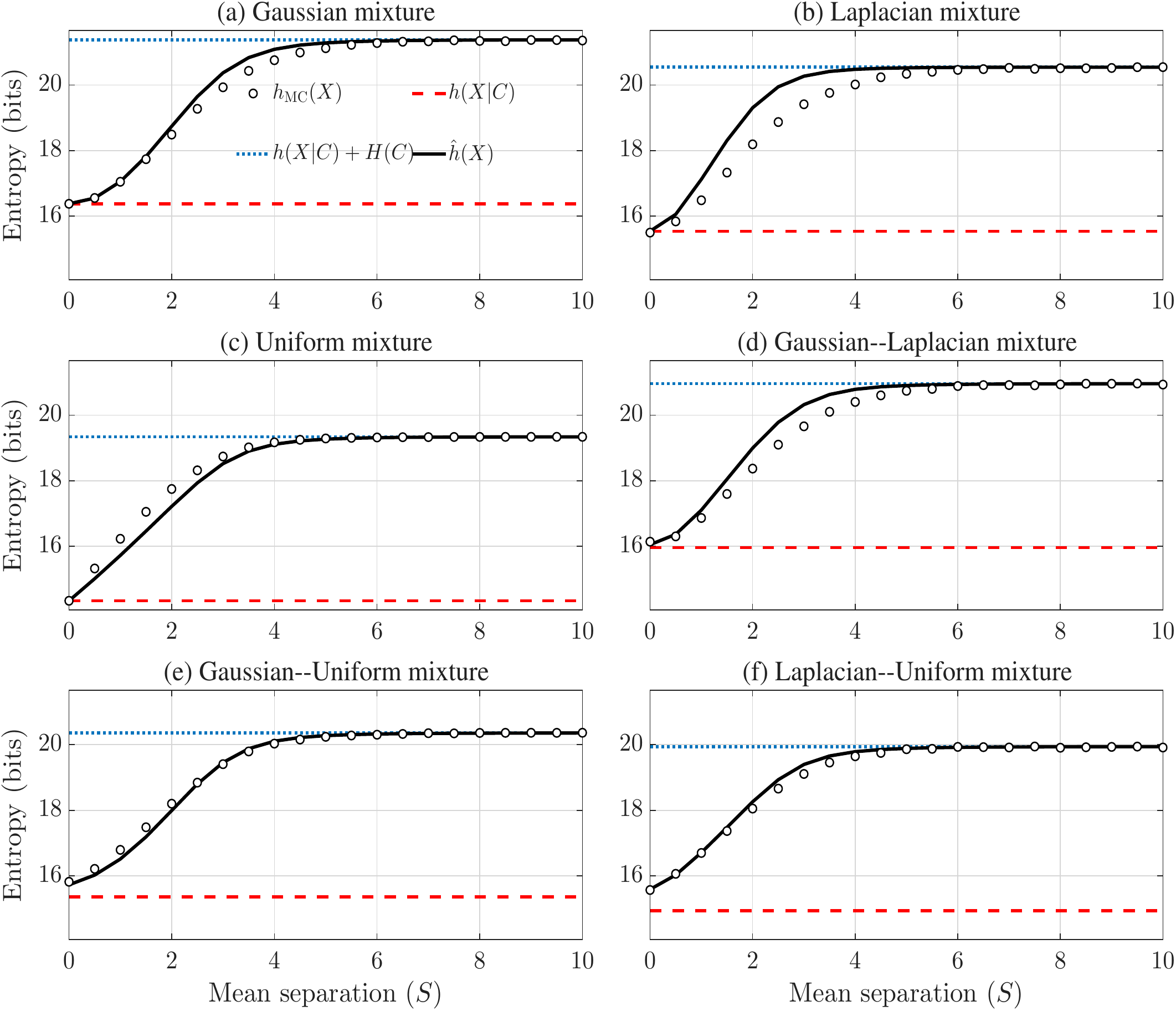}
\caption{Comparison of $h(X)$ for various mixture distributions when $(n,K)=(8,8)$: MC entropy, sandwich bounds, and clipped approximation versus mean separation.}
\label{fig:rep_case1_n8K8}
\end{figure}

\begin{figure}[t]
\centering
\includegraphics[width=0.7\linewidth]{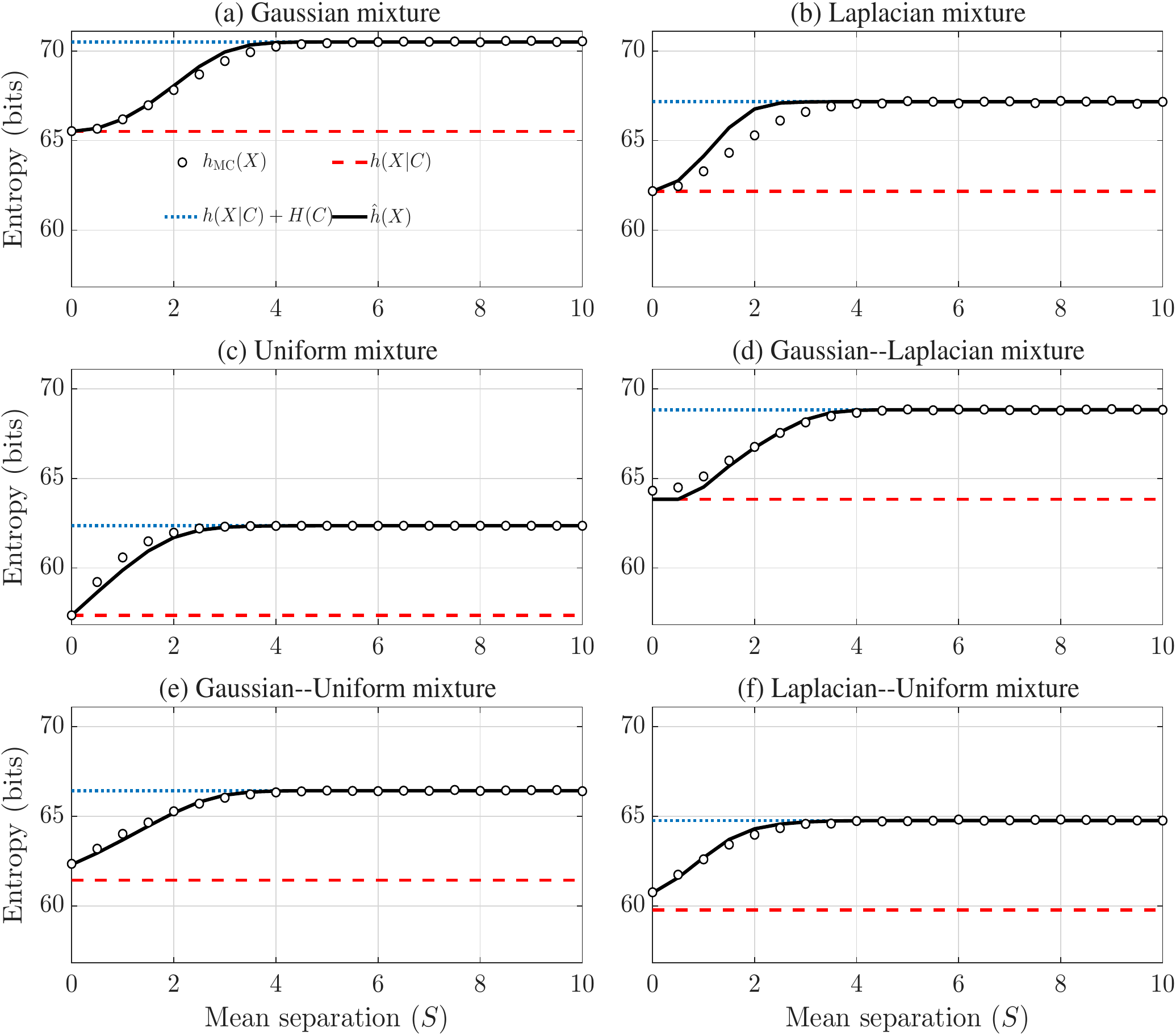}
\caption{Comparison of $h(X)$ for various mixture distributions when $(n,K)=(32,32)$: MC entropy, sandwich bounds, and clipped approximation versus mean separation.}
\label{fig:rep_case1_n32K32}
\end{figure}

\subsection{Six Mixture Families}
\label{subsec:num_interp_indep}



{\bf Effect of $K$:}
The vertical gap between the two bounds is fixed at $H(C)=\logtwo K$ under equal weights.
Thus increasing $K$ increases the maximum possible entropy lift caused by mixing.
This is why, when $K$ is larger (e.g., Fig.~\ref{fig:rep_case1_n32K32} with $K=32$),
the admissible interval is visibly taller than in cases with smaller $K$.
Operationally, $\logtwo K$ is the number of bits needed to describe the label if it were fully decodable.

{\bf Effect of $n$:}
A consistent trend across Fig.~\ref{fig:rep_case1_n2K8}--Fig.~\ref{fig:rep_case1_n32K32} is that the transition
from the lower plateau to the upper plateau occurs over a narrower range of $S$ as $n$ grows.
This is an overlap-multiplication effect:
for product-form families (Laplacian and uniform boxes) the overlap factors across coordinates, and even for
Gaussians the exponent involves an $n$-dimensional quadratic form.
As a result, moderate per-coordinate separations compound in high dimension, making the label effectively decodable
at smaller $S$.

{\bf Effect of shape:} Even though all families are variance matched, their overlap geometry differs:
\begin{itemize} 
\item Uniform mixtures (compact support)often saturate earlier in $S$ because once boxes become largely disjoint,
cross-overlaps collapse rapidly.
\item Gaussian mixtures (light tails) exhibit a smooth, exponential-in-distance overlap decay, producing a gradual
      approach to the upper bound.
\item Laplacian mixtures (heavier tails) tend to sustain overlap longer at moderate separations, so the rise toward
      the upper bound can be more gradual than in the Gaussian case (variance matched).
\end{itemize}
These differences are precisely what the overlap-based proxy is designed to capture: the only nontrivial input to
$h_{\sf L}(X)$ is the overlap matrix $\{z_{c,d}\}$, which is sensitive to tail behavior and support geometry.

{\bf Shape decodability in hybrids even at $S=0$:}
The hybrid panels (G--L, G--U, L--U) illustrate a subtle but important point: mean coincidence ($S=0$) does not necessarily imply $I(X;C)=0$. When component families differ, the likelihood shapes differ, so a single observation can still be informative about which component generated it. This shape-driven decodability becomes more visible as $n$ increases because repeated coordinates provide more evidence (concentration of log-likelihood ratios). Consequently, hybrid mixtures can exhibit a nontrivial entropy lift above $h(X\mid C)$ already at $S=0$ as shown in Fig. \ref{fig:rep_case1_n32K32}-(d), i.e, the G-L mixture case.

{\bf Accuracy and robustness of $\hat h(X)$:}
Across all panels, the proposed approximation $\hat h(X)$ tracks the Monte Carlo estimate well at both endpoints:
\begin{itemize}[leftmargin=*, itemsep=2pt]
\item \textbf{Small separation:} when overlap is large, $I(X;C)$ is small, and the offset calibration makes
      $h_{\sf L}(X)+\bar\Delta$ agree with the correct normalization.
\item \textbf{Large separation:} when overlap is small, $I(X;C)\approx H(C)$, so $h(X)\approx h(X\mid C)+H(C)$ and
      $\hat h(X)$ approaches the upper plateau.
\end{itemize}
In intermediate regimes, clipping is the safety mechanism: it enforces \eqref{eq:num_sandwich} exactly and prevents
numerically or structurally induced overshoot/undershoot that would contradict basic information identities.
\begin{figure}[t]
\centering
\includegraphics[width=0.7\linewidth]{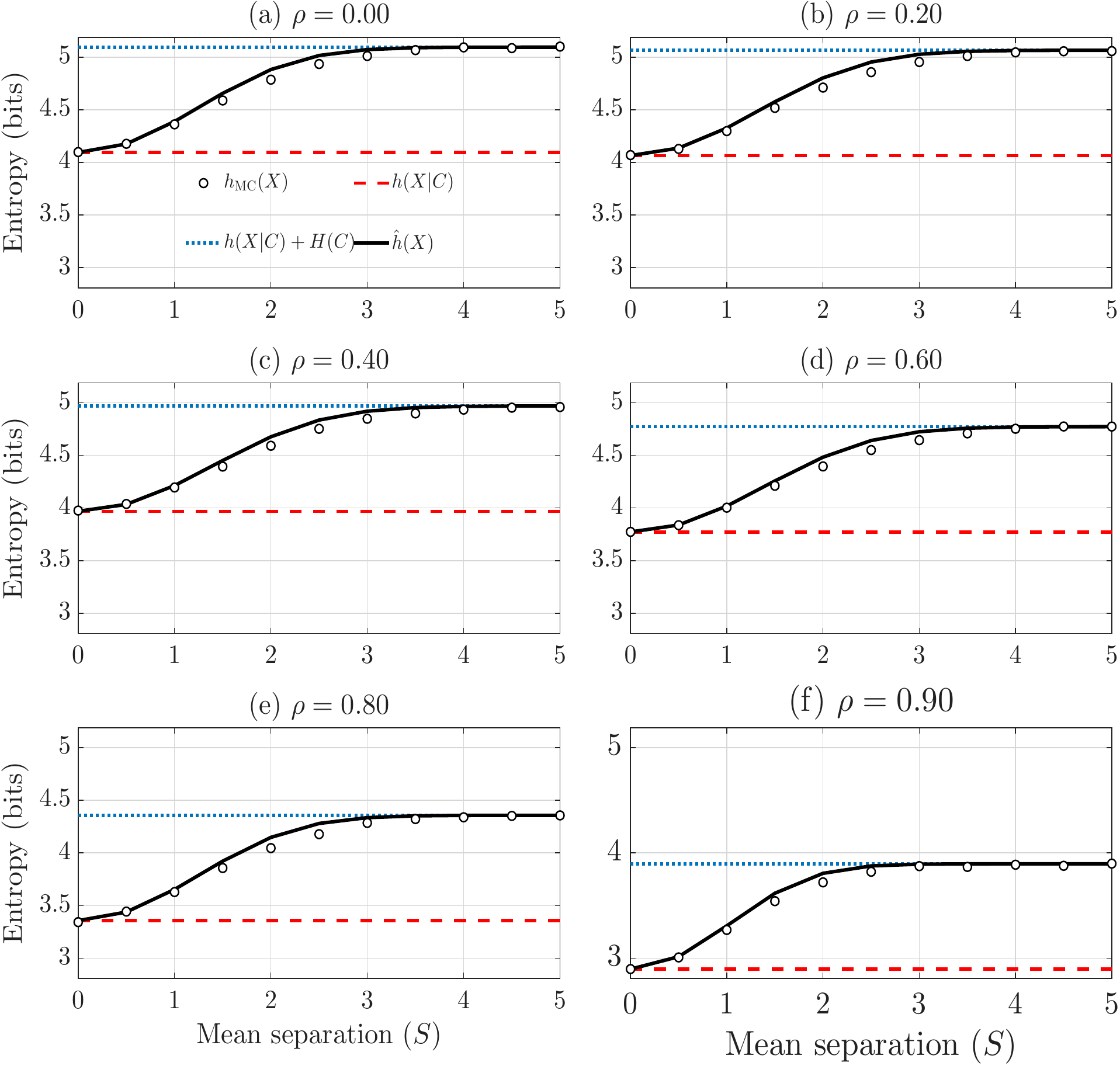}
\caption{Correlated Gaussian mixtures when $(n,K)=(2,2)$ and correlation $\rho\in\{0,0.2,0.4,0.6,0.8,0.9\}$: MC entropy, sandwich bounds, and clipped approximation versus mean separation.}
\label{fig:correlated1}
\end{figure}

\begin{figure}[t]
\centering
\includegraphics[width=0.7\linewidth]{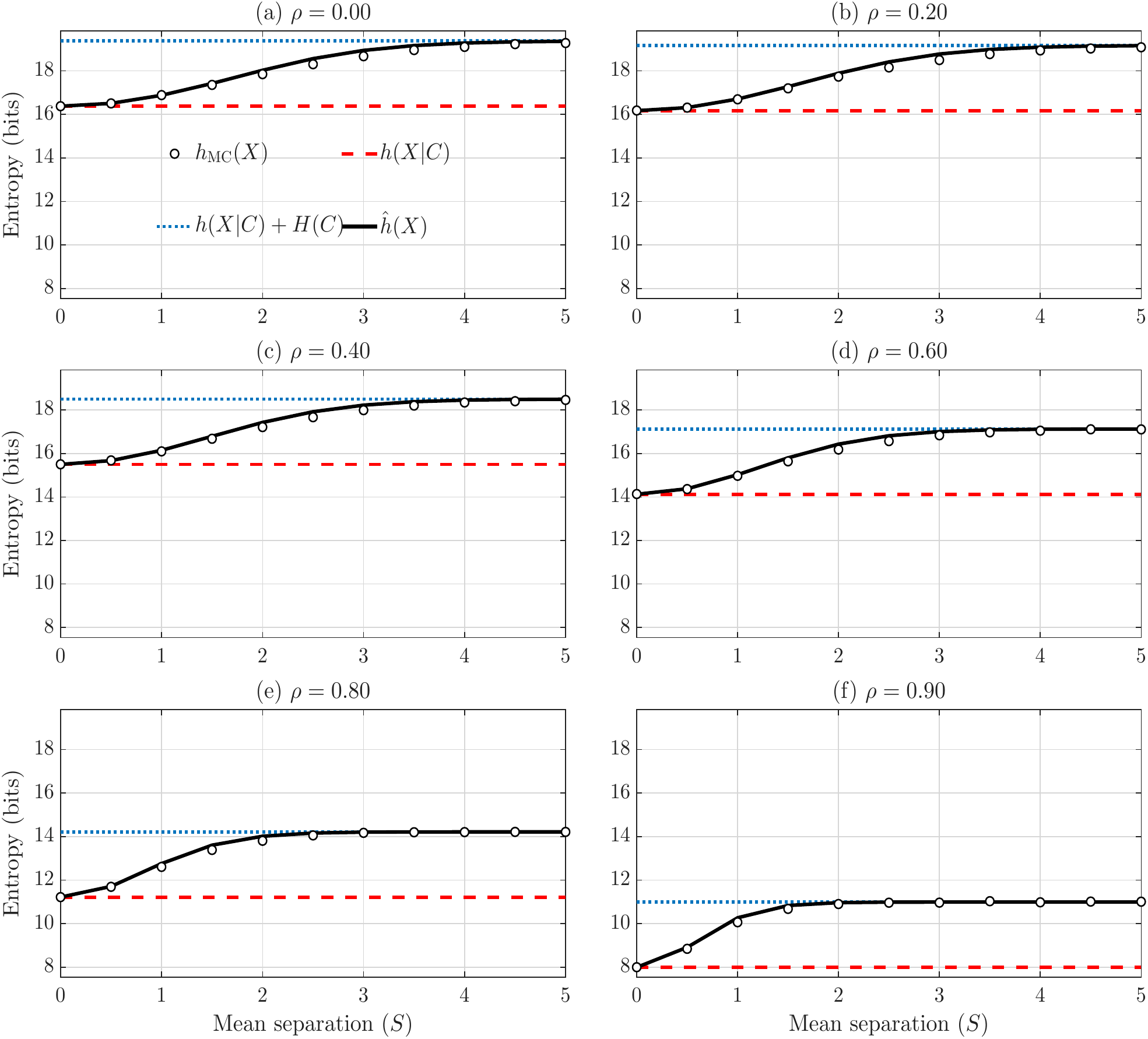}
\caption{Correlated Gaussian mixtures when $(n,K)=(8,8)$ and correlation $\rho\in\{0,0.2,0.4,0.6,0.8,0.9\}$: MC entropy, sandwich bounds, and clipped approximation versus mean separation.}
\label{fig:correlated2}
\end{figure}
\subsection{Correlated Gaussian mixtures}
\label{subsec:num_interp_corr}

We additionally test correlated Gaussian mixtures,
\begin{align}
p_X(x)
=\frac{1}{K}\sum_{c=1}^K \mathcal{N}\!\big(x;\mu_c(S),\Sigma_\rho\big),
\label{eq:num_corr_model}
\end{align}
where $\Sigma_\rho$ is an first-order autoregressive type covariance with correlation parameter $\rho$ (as used to generate the figures).
Fig.~\ref{fig:correlated1} reports $(n,K)=(2,2)$ and Fig.~\ref{fig:correlated2} reports $(n,K)=(8,8)$, each for
$\rho\in\{0,0.2,0.4,0.6,0.8,0.9\}$.

{\bf Effect of correlation:}
Correlation changes the component entropy through $|\Sigma_\rho|$:
\begin{align}
h(X\mid C)=\frac12\logtwo\!\big((2\pi e)^n|\Sigma_\rho|\big),
\label{eq:num_corr_hxgc}
\end{align}
so increasing $\rho$ generally decreases $|\Sigma_\rho|$ (more dependence, less volume) and shifts both sandwich
bounds downward by the same amount. This is visible as a vertical shift of both reference lines as $\rho$ varies.

For Gaussians, overlaps depend on Mahalanobis distance under $(\Sigma_\rho+\Sigma_\rho)^{-1}=(2\Sigma_\rho)^{-1}$.
Thus, even at fixed Euclidean mean scaling $S$, the \emph{effective} separation is governed by the geometry induced by
$\Sigma_\rho^{-1}$.
As $\rho$ increases and $\Sigma_\rho$ becomes more anisotropic, separation along low-variance directions becomes more
informative, and the transition toward the upper bound can occur at smaller $S$.
This is the correlated analogue of the high-dimensional sharpening effect: correlation reshapes the space so that
some directions provide more discriminative power.

{\bf Tightness of our approximation:} 
Theorem~\ref{thm:tight_approx} only requires that overlaps and component entropies be computable.
For correlated Gaussians, both remain closed form: $h(X\mid C)$ follows from \eqref{eq:num_corr_hxgc} and
$z_{c,d}$ follows from the Gaussian overlap formula with $\Sigma_c=\Sigma_d=\Sigma_\rho$.
Therefore $\hat h(X)$ remains fully deterministic and continues to respect the sandwich constraints by construction,
while tracking the Monte Carlo estimate over the full sweep of $(S,\rho)$.

\section{Conclusion}
\label{sec:conclusion}

This paper takes a mixture model and treats it as a channel: the discrete label $C$ goes in, the continuous observation
$X$ comes out.  That viewpoint immediately turns a difficult analytic object—the entropy of a log-sum density—into a
clean accounting identity,
\[
h(X)=h(X\mid C)+I(X;C),
\]
and with it, the universal label-sandwich bracket
\[
h(X\mid C)\ \le\ h(X)\ \le\ h(X\mid C)+H(C).
\]
The message is simple: the only part of the mixture entropy that is not already in closed form is the mutual
information term.  Everything interesting is happening in $I(X;C)$, i.e., in how much the data reveal the label, or
equivalently, in the overlap geometry of the components.

The main contribution is then to replace that implicit, posterior-based information term by an explicit surrogate
built from pairwise overlaps.  Jensen’s inequality produces a computable overlap functional $h_{\sf L}(X)$, and the
collision–Shannon gap supplies a family-dependent offset $\bar\Delta$ that calibrates this surrogate at the complete-overlap
endpoint.  A final clipping step enforces the only two inequalities that can never be violated, regardless of dimension,
number of components, or parameter regime.  The result is a closed-form approximation $\hat h(X)$ that is anchored at the
two information-theoretic extremes—where the structure collapses and the answer is known—and remains admissible everywhere
in between.  Operationally, the framework reduces mixture-entropy evaluation to a single question: can we compute overlaps
$z_{c,d}$?  For the time-honored families studied here (Gaussian, factorized Laplacian, uniform on sets, and common hybrids),
the answer is yes, yielding practical formulas with $O(K^2 n)$ complexity in high dimension.

The simulations reinforce the theory in the way one would hope: as the mean-separation parameter increases, the Monte Carlo
entropy moves smoothly from $h(X\mid C)$ toward $h(X\mid C)+H(C)$, and the proposed approximation tracks that transition while
never leaving the admissible interval.  Dimension sharpens the transition because overlaps compound across coordinates; compact
support saturates earlier; heavier tails sustain overlap longer; and in hybrid mixtures, shape differences can make the label
partially decodable even when means coincide.  Correlated Gaussian experiments show that the same story survives once Euclidean
distance is replaced by the geometry induced by covariance: overlaps remain closed form, and the approximation remains stable.

More broadly, the overlap-based lens offers a useful design principle. If we are building a model, a codebook, or a clustering
representation and we can control overlaps, then we can control the information revealed about the label, hence control the
mixture entropy.  That is the bridge this paper aims to build: from closed-form overlap geometry to entropy and mutual-information
estimates that are both computable and information-theoretically consistent.

\bibliographystyle{IEEEtran}
\bibliography{ref}

\end{document}